\documentclass{article}

\usepackage{arxiv}
\usepackage[T1]{fontenc}
\usepackage[utf8]{inputenc}
\usepackage{tikz}
\usetikzlibrary{arrows,arrows.meta}
\usepackage{xparse}
  \usepackage{amsmath}
\usepackage{amsthm}
  \usepackage{amssymb}
  \usepackage{hyperref}

  \hypersetup{
    colorlinks,
    linkcolor={blue!70!black},
    citecolor={blue!70!black},
    urlcolor={black}
}
\usepackage{svn-multi}

\svnid{$Id: disc.tex 43301 2020-03-04 17:05:01Z ti5vt $}

\usepackage[export]{adjustbox}

\usepackage[ruled,vlined]{algorithm2e}
\SetKwFor{round}{\hspace*{-0.78ex}}{}{end}
\SetKwFunction{randomColor}{randomColor}
\SetKw{KwRet}{return}
\SetKw{Return}{return}
\SetKwInOut{Input}{input}

\graphicspath{{./figures/}}

\newcommand{\pr}{\textit{parity}\xspace}
\newcommand{\mess}{\textit{messTbl}\xspace}

\newcommand{\false}{\mathit{false}}
\newcommand{\true}{\mathit{true}}
\newcommand{\abs}[1]{{\left\lvert #1 \right\rvert}}

\newcommand{\suchthat}{\;\ifnum\currentgrouptype=16 \middle\fi|\;}

\newcommand{\B}{\ensuremath{{\cal B}}\xspace}
\newcommand{\G}{\ensuremath{{\cal G}}\xspace}

\newtheorem{theorem}{Theorem}

\newtheorem{lemma}[theorem]{Lemma}

\makeatletter
\newcommand*{\Xbar}{}%
\DeclareRobustCommand*{\Xbar}{%
  \mathpalette\@Xbar{}%
}
\newcommand*{\@Xbar}[2]{%
  \sbox0{$#1\mathrm{X}\m@th$}%
  \sbox2{$#1X\m@th$}%
  \rlap{%
    \hbox to\wd2{%
      \hfill
      $\overline{%
        \vrule width 0pt height\ht0 %
        \kern\wd0 %
      }$%
    }%
  }%
  \copy2 %
}
\makeatother

\DeclareRobustCommand{\doublearrow}{%
  \tikz{%
    \draw[-{Latex[length=1.4mm,width=1.5mm] Latex[length=1.4mm,width=1.5mm]}] (0.0,0.0) -- (0.45,0.0);
  }%
  \xspace%
}

\RequirePackage{mleftright} 
\RenewDocumentCommand{\Pr}{ s o o m }{%
  \IfBooleanTF{#1}{%
    \IfNoValueTF{#2}{%
      \specrandop{Pr}[#4]%
    }{%
      \IfNoValueTF{#3}{%
        \specrandop{Pr}\mathopen{#2[}#4\mathclose{#2]}%
      }{%
        \specrandop{Pr}_{#3}\mathopen{#2[}#4\mathclose{#2]}%
      }
    }%
  }{%
    \IfNoValueTF{#2}{%
      \specrandop{Pr}\mleft[#4\mright]%
    }{%
      \specrandop{Pr}_{#2}\mleft[#4\mright]%
    }
  }%
}

\usepackage{fontawesome5}
\usepackage{listings}
\usepackage{mathtools}
\usepackage{tuhhcolor}

\lstdefinestyle{daStyle}{
  language         = C++,
  basicstyle       = \ttfamily\small,
  commentstyle     = \color{tuhh_darkgray},
  keywordstyle     = \bfseries,
  aboveskip        = 1pt,
  belowskip        = 0pt,
  fontadjust       = true,
  columns          = fullflexible,
  keepspaces       = true,
  tabsize          = 3,
  showstringspaces = false,
  mathescape       = true,
  escapechar       = €,
  morekeywords     = {concurrently,send,start,to,do,down,receive,id,while,repeat
,foreach,until,init,compute,let,wait,null,function},
  moredelim        =**[is][\color{red}]{§}{§},
  moredelim        =**[is][\sffamily\slshape]{`}{`},
  literate         = {Ä}{{\"A}}1
                     {Ö}{{\"O}}1
                     {Ü}{{\"U}}1
   	                 {ä}{{\"a}}1
			               {ö}{{\"o}}1
			               {ü}{{\"u}}1
			               {ß}{{\ss{}}}1
                     {_m}{\textsubscript{m}}2
                     {_p}{\textsubscript{p}}2
                     {_u}{\textsubscript{u}}2
                     {_v}{\textsubscript{v}}2
                     {_w}{\textsubscript{w}}2
                     {_t}{\textsubscript{t}}2
                     {_p}{\textsubscript{p}}2
                     {_id}{id}3
                     {_or}{or}3
                     {_for}{for}4
                     {_not}{not}4
                     {_and}{and}4
                     {_compute}{compute}8
                     {_receive}{receive}8
                     {_start}{start}5
                     {_wait}{wait}5
                     {_until}{until}6
                     {_and}{and}4
                     {_send}{send}5
                     {_new}{new}4
                     {_not}{not}4
                     {_to}{to}3
                     {_do}{do}3
                     {_let}{let}4
                     {_using}{using}6
                     {:}{$\;:\;$}1
                     {*}{{$\ast$}}1
                     {,}{,$\,$}1
                     {+}{{$\,+\,$}}1
                     {*}{{$\,\ast\,$}}1
                     {\%}{{$\,\%\,$}}1
                     {-}{{$\,-\,$}}1
                     {>}{{$\;>\;$}}1
                     {≥}{{$\;\ge\;$}}1
                     {>=}{{$\;\ge\;$}}2
                     {<}{{$\;<\;$}}1
                     {≤}{{$\;\le\;$}}1
                     {<=}{{$\;\le\;$}}2
                     {=}{{$\;=\;$}}1
                     {...}{{$\,\dots$}}3
                     {≠}{{$\;\neq\;$}}1
                     {/=}{{$\;\neq\;$}}2
                     {:=}{{$\;\coloneqq\;$}}2
                     {\\}{{$\,\backslash\,$}}1
                     {|}{{$\,\mid\,$}}1
                     {_NN}{{$\mathbb{N}$}}3
                     {N+(}{N\textsuperscript{+}$\!$(}3
                     {_in}{{$\!\!\in\!\!$}}2
                     {_notin}{{$\!\!\not\in\!\!$}}2
                     {__in}{\textsubscript{in}}4
                     {__out}{\textsubscript{out}}5
                     {_cup}{{$\!\cup\!$}}4
                     {_omega}{{$\omega$}}6
                     {_Delta}{{$\Delta$}}6
                     {_exists}{{$\Exists\!\!\!$}}7
                     {_forall}{{$\Forall\!\!\!$}}7
                     {_empty}{{$\emptyset$}}6
                     {_infty}{{$\infty$}}6
                     {_bot}{{$\bot$}}4
                     {_\{}{{$\{$}}2
                     {_\}}{{$\}$}}2
                     {_:}{$:$}2
                     {/\\}{{$\!\wedge\!$}}2
                     {\\/}{{$\!\vee\!$}}2
                     {<<}{{$\langle$}}2
                     {>>}{{$\rangle$}}2
                     {s.t.}{{\textbf{s.t.}}}4
                     {_unvis}{\textsubscript{unvis}}6
                     {_0}{\textsubscript{0}}1
                     {_1}{\textsubscript{1}}1
                     {_2}{\textsubscript{2}}1
                     {_3}{\textsubscript{3}}1
                     {_4}{\textsubscript{4}}1
                     {_5}{\textsubscript{5}}1
                     {_6}{\textsubscript{6}}1
                     {_7}{\textsubscript{7}}1
                     {_8}{\textsubscript{8}}1
                     {_9}{\textsubscript{9}}1
                     {_notp}{\textsubscript{$\overline{\texttt{p}}$}}5
                     {not_p}{$\overline{\texttt{p}}$}5
                     {not_parity}{$\overline{\texttt{parity}}$}5
                     {++}{{++}}2
                     {--}{{--}}2
                     {init:}{{\textbf{\textcolor{\cmdColor}{init}}:}}5
                     {start:}{{\textbf{\textcolor{\cmdColor}{start}}:}}6
                     {compute:}{{\textbf{\textcolor{\cmdColor}{compute}}:}}8
                     {Inevery:}{{\textbf{\textcolor{\cmdColor}{In every round do}}:}}8
                     {receive(}{{\textbf{\textcolor{\cmdColor}{receive}}(}}8
                     {_receive(}{{\textbf{receive}(}}9
}

\def\cmdColor{tuhh_darkred}

\newcommand{\fontmathtext}[1]{\mathsf{#1}}%
\DeclareRobustCommand{\alorithm}[1]{{\ensuremath{{\cal A}_\fontmathtext{#1}}}\xspace}

\DeclareRobustCommand{\SynAF}{\alorithm{AF}}
\DeclareRobustCommand{\SynAFI}{\alorithm{AFI}}
\DeclareRobustCommand{\SynAFIF}{\alorithm{AFI^F}}
\DeclareRobustCommand{\SynAFIS}{\alorithm{AFI^S}}

\title{Synchronous Concurrent Broadcasts for Intermittent Channels
  with Bounded Capacities}
\author{%
  \href{https://orcid.org/0000-0001-9964-8816}{%
    \includegraphics[height=0.8em]{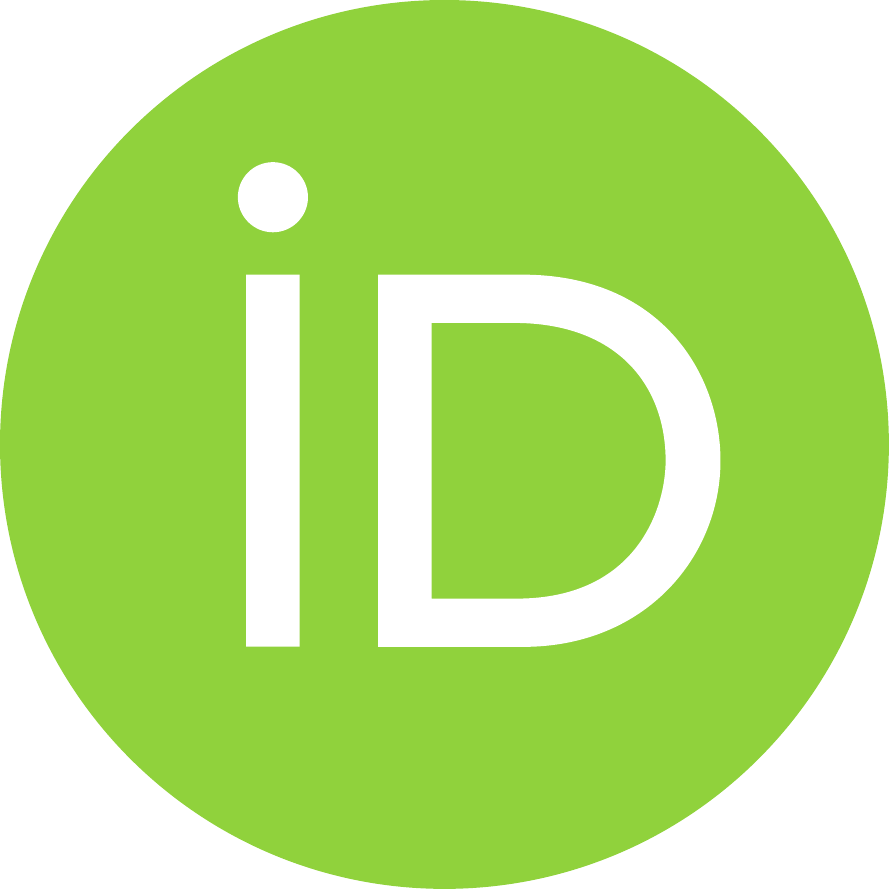}%
    \hspace{1mm}%
    Volker Turau%
  } \\
  Institute of Telematics \\
  Hamburg University of Technology \\
  21073 Hamburg, Germany \\
  \texttt{turau@tuhh.de}
}

\begin{document}

\maketitle

\begin{abstract}
In this work we extend the recently proposed synchronous broadcast
  algorithm {\em amnesiac flooding} to the case of intermittent
  communication channels. In amnesiac flooding a node forwards a
  received message in the subsequent round. There are several reasons
  that render an immediate forward of a message impossible: Higher
  priority traffic, overloaded channels, etc. We show that postponing
  the forwarding for one or more rounds prevents termination. Our
  extension overcomes this shortcoming while retaining the advantages
  of the algorithm: Nodes don't need to memorize the reception of a
  message to guarantee termination and messages are sent at most twice
  per edge. This extension allows to solve more general broadcast
  tasks such as multi-source broadcasts and concurrent broadcasts for
  systems with bounded channel capacities.
\end{abstract}


\section{Introduction}

Broadcasting is the task of delivering a message from one network node
to all other nodes. Broadcast algorithms constitute a fundamental
component of many distributed systems and are often used as
subroutines in more complex algorithms. There are numberless
applications of broadcast. Demers et al.\ discuss the 
maintenance of a database replicated at many sites in a large
corporate network \cite{Demers:1987}. Each database update can be
injected at various nodes, and these updates must be propagated to all
nodes in the network. The replica become fully consistent only when
all updating activity has stopped and the system has become quiescent.
The efficiency of the broadcasting algorithm determines the rate of
updates the system can handle.

A common broadcasting algorithm is flooding. The originator $v_0$ of a
message $m$ forwards $m$ to all neighbors and when a node receives $m$
for the first time, it sends it to all its neighbors in the
communication graph $G(V,E)$. Flooding uses $2\abs{E}$ messages and
terminates after at most $\epsilon_G(v_0)+1$ rounds, $\epsilon_G(v_0)$
denotes the maximal distance of $v_0$ to any other node. In this form
flooding is a stateful algorithm, it requires each node to keep a
record of already forwarded messages. This requires storage per node
in the order of the number of broadcasted messages. Since nodes are
unaware of the termination of the broadcast, these records have to be
stored for an unknown time.

For synchronous distributed systems stateless broadcasting algorithms
are known. Hussak and Trehan proposed {\em amnesiac flooding} (\SynAF)
\cite{Hussak:2020}. Every time a node receives message $m$, it
forwards it to those neighbors from which it didn't receive $m$ in the
current round. In contrast to classic flooding, a node may forward a
message twice. Surprisingly amnesiac flooding terminates and each
message is sent at most twice per edge. Crucial for the termination of
\SynAF is that the forwarding of messages is always performed in the
round immediately following the reception. We show in
Sec.~\ref{sec:afi} that algorithm \SynAF no longer terminates when
message forwarding is suspended for some rounds. There can be several
reasons for suspending forwarding, when traffic with a priority higher
than broadcast has to be handled, or when the capacity of a
communication channel is exhausted due to several concurrent
broadcasts. Surprisingly it requires only a simple extension to make
\SynAF to work correctly despite a limited number of suspensions. Our
first contribution is the extended algorithm \SynAFI described in
Sec.~\ref{sec:afi}.

Our first result enables us to prove that algorithm \SynAF is also
correct for multi-source broadcasting, i.e., several nodes broadcast
the same message $m$ in different rounds, provided a broadcast of $m$
is invoked before $m$ reaches the invoking node from another
broadcast. In Sec.~\ref{sec:multi-source} we prove that in this case
\SynAF delivers $m$ after at most $Diam(G)$ rounds and forwards $m$ at
most $2\abs{E}$ times. If the communication channel is unavailable $f$
times then \SynAFI delivers $m$ after at most $Diam(G)+2f$ rounds, $m$
is still forwarded at most $2\abs{E}$ times.

While algorithm \SynAFI is of interest on its own, it can also be used
to solve the general task of multi-message broadcast in systems with
bounded channel capacities. Multi-message broadcast means that
multiple nodes initiate broadcasts of different messages, even when
broadcasts from previous initiations have not yet terminated. If
channel capacities are bounded, nodes can forward only a limited
number of messages per round. Bounded channel capacities occur in
communication systems utilizing TDMA, where communication is performed
in fixed length slots and therefore only $b$ messages can be sent in
one round. If more than $b$ messages are in the sending queue, then
the forwarding of some messages has to be postponed for at least one
round. In Sec.~\ref{sec:dynamic} we present two algorithms \SynAFIS
and \SynAFIF for this task. The advantage of these algorithms is that
compared to classic flooding besides the unavoidable message buffer no
state information has to be maintained. Thm.~\ref{thm:main} summarizes
our third contribution.

\begin{theorem}\label{thm:main}
  Let $\cal S$ be a sequence of message broadcasts (identical or
  different messages) by the nodes of a graph $G(V,E)$ in arbitrary
  rounds under the restriction that a broadcast of message $m$ is
  invoked before $m$ reaches the invoking node from a broadcast of
  another node. If in each round each node can send at most $b$
  messages to each neighbor algorithm \SynAFIF eventually terminates
  and delivers each message of $\cal S$. Nodes don’t need to memorize
  the reception of a message. If $G$ is bipartite each message is
  forwarded $\abs{E}$ times, otherwise $2\abs{E}$ times.
\end{theorem}

\section{State of the Art}
\label{sec:state-art}
Broadcasting as a service in distributed systems can be realized in
two ways: Either by using a pre-constructed structure such as a
spanning tree or by performing the broadcast each time from scratch.
In the first case a broadcast can be performed with $n-1$ messages. In
the second case a broadcast can be realized by $2(n-1)$ messages by
traversing the graph in a DFS style and carrying the identifiers of the
visited nodes along with the messages. This requires messages that
store up to $n$ node identifiers. If the message size is
restricted to $o(n)$ and only a fixed number of messages can be
sent per round per link then each deterministic broadcast algorithm has
message complexity $\Omega(\abs{E})$, Thm.~23.3.6 \cite{Peleg:2000}.
For a detailed analysis of broadcast algorithms we refer to Sec.~23 of
\cite{Peleg:2000}.

In this work we focus on broadcast algorithms that do not rely on a
pre-con\-structed structure and use limited communication channels.
The most basic algorithm of this category is flooding as described
above. Flooding uses $2\abs{E}$ messages and terminates after at most
$\epsilon_G(v_0)+1$ rounds, these bounds hold in the synchronous and
asynchronous model \cite{Peleg:2000}. It requires each node to
maintain for each message a record that the message has been
forwarded. These records have to be kept for an unknown time. This
requires storage per node proportional to the number of disseminated
messages. Amnesiac flooding \SynAF overcomes this limitation in
synchronous systems and is thus stateless \cite{Hussak:2020}. \SynAF
delivers a broadcasted message twice to each node. Thus, we have to
distinguish between delivery and termination time. \SynAF delivers a
message (resp.\ terminates) for an initiator $v_0$ on any finite graph
in at most $\epsilon_G(v_0)$ (resp.\ $\epsilon_G(v_0) + Diam(G)+1$)
rounds, where $Diam(G)$ is the diameter of $G$. The termination time
compared to standard flooding increases almost by a factor of $2$.
Amnesiac flooding was also analyzed for sets of initiators
\cite{Turau:2020b}. A stateless broadcasting algorithm with the same
time complexity as classic flooding has recently been proposed in
\cite{Turau:2020a}.

A problem related to broadcast is rumor spreading. It describes the
dissemination of information in networks through pairwise
interactions. A simple model for rumor spreading is that in each
round, each node that knows the rumor, forwards it to a randomly
chosen neighbor. For many topologies, this strategy is a very
efficient way to spread a rumor. With high probability the rumor is
received by all vertices in time $\Theta(\log n)$, if the graph is a
complete graph or a hypercube \cite{Frieze:1985,Feige:1990}. New
results about rumor spreading can be found in \cite{Mocquard:2020}.

Intermittent channel availability is no issue for classic flooding and
thus has not been considered. Broadcasting in distributed systems with
bounded channel capacities has received little attention. Hussak et
al.\ consider a model where each node can send a single message per
edge per round \cite{Hussak:2020c}. They propose variants of amnesiac
flooding to handle the case of many nodes invoking broadcasts of
different messages in different rounds. They show that their
algorithms terminate, but message delivery to all nodes is only
guaranteed in the special case that a single node broadcasts different
messages. Our work is more general and uses a different approach.

Raynal et al.\ present a broadcast algorithm suited for dynamic
systems where links can appear and disappear \cite{Raynal:2014}. Some
algorithms of \cite{Hussak:2020c} also maintain their properties in
case edges or nodes disappear over time. Casteigts et al.\ analyze
broadcasting with termination detection in time-varying graphs
\cite{Casteigts:2010}. They prove that the solvability and complexity
of this problem varies with the metric considered, as well as with the
type of a priori knowledge available to nodes.

\section{Notation and Model}
In this work $G(V,E)$ denotes a finite, connected, undirected graph
with $n=|V|$. Let $v,u\in V$, $d_G(v,u)$ denotes the distance between
$v$ and $u$ in $G$, $N(v)$ the set of neighbors and $\epsilon_G(v)$
the {\em eccentricity} of $v$ in $G$, i.e., the greatest distance
between $v$ and any other node in $G$. $Diam(G)$ denotes the maximum
eccentricity of any node of $G$. An edge $(u,w)\in E$ is called a {\em
  cross edge} with respect to a node $v_0$ if
$d_G(v_0,u) = d_G(v_0,w)$. $\Delta$ denotes the maximal node degree in
$G$. Each node has a unique id and is aware of the ids of its
neighbors but does not have any knowledge about graph parameters such
as the number of nodes or diameter.

The goal of a broadcasting algorithm is to disseminate a message
created by a node to all nodes of the network. Messages are assumed to
be distinguishable, each having unique id. No message is lost in
transit. A broadcast is said to {\em terminate} when all network
events (message sends/receives) that were caused by that broadcast
have ceased. A broadcast message is said to have been {\em delivered},
if it has been received by all the nodes in the network.

In this paper we consider synchronous distributed systems, i.e.,
algorithms are executed in rounds of fixed length and all messages
sent by all nodes in a particular round are received and processed in
the next round. In Sec.~\ref{sec:dynamic} we assume that in each round
each node can only send a constant number $b$ of messages to a subset
of its neighbors. This can be realized by a network-level broadcast,
where each message contains the identifiers of the receivers. This
requires $O(\Delta\log n)$ bits in each messages. Besides this, each
message has just enough space to contain the information to be
disseminated. In particular two messages cannot be aggregated into
one.

\section{Handling Intermittent Channels}\label{sec:afi}
In this section we extend \SynAF so that it operates correctly with
intermittent channel availabilities. Alg.~\ref{alg:af} recaps the
details of amnesiac flooding \SynAF as described in
\cite{Hussak:2020}. A node that wants to flood a message $m$ sends $m$
to all neighbors. Every time a node receives $m$, it forwards it to
those neighbors from which it didn't receive $m$ in the current round.
The code in Alg.~\ref{alg:af} shows the handling of a single message
$m$. If several messages are broadcasted concurrently, each requires
its own set $M$.

 \SetProgSty{}
 \begin{algorithm}
  \Input{A graph $G=(V,E)$, a subset $S$ of $V$, and a message $m$.}
  \BlankLine
  In round $1$ each node $v\in S$ sends message $m$ to each neighbor
  in $G$\;
   \round{Each node $v$ executes in every round $i>1$}{
   $M := N(v)$\;
   \ForEach{receive$(w,m)$}{$M := M \setminus \{w\}$\;}
   \If{$M \not=N(v)$} {
   \lForAll{$u\in M$}{send($u,m$)}}
 }
\caption{Algorithm \SynAF distributes a message $m$ in the
  graph $G$}\label{alg:af}
\end{algorithm}

An attempt to handle channel unavailabilities is to postpone the
sending of some messages to the next round when the channel is again
available. Messages received in the mean time are treated
as before, the senders are inserted into $M$. Unfortunately, this
modification of \SynAF may not terminate. Fig.~\ref{fig:unlimited}
presents an illustrative example. In the graph depicted in the top
left node $v_0$ broadcasts a message $m$ in round $0$. Suppose that
node $v_2$ (resp.\ $v_3$) cannot send messages in rounds $2,3$ and $4$
(resp.\ in round $2$). We show that forwarding messages in the first
available round may prevent termination. In the first round $v_0$
sends $m$ to $v_1,v_2$ and $v_3$. In round $2$ nodes $v_2$ and $v_3$
cannot forward $m$ and postpone the sending. Node $v_3$ postpones this
to round $3$. In this round $v_2$ also receives a message from $v_1$.
In rounds $3$ and $4$ node $v_2$ in addition receives a message from
node $v_5$. These three events cannot be handled immediately and are
also postponed. In round $5$ the channel becomes available for node
$v_2$, but in the meantime $v_2$ has received a message from all its
neighbors and thus \SynAF will not send $m$ to any of $v_2$'s
neighbors. From this round on the channel is continuously available
and thus \SynAF can be executed in its original form. In round $9$ the
algorithm reaches the same configuration as in round $5$. Thus, the
algorithm does not terminate.

\begin{figure}[h]%
  \hfill
    \includegraphics[scale=0.9]{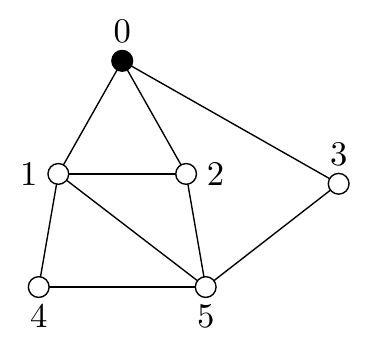}
  \hfill
    \includegraphics[scale=0.9]{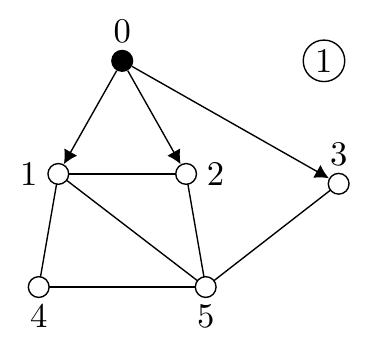}
  \hfill
    \includegraphics[scale=0.9]{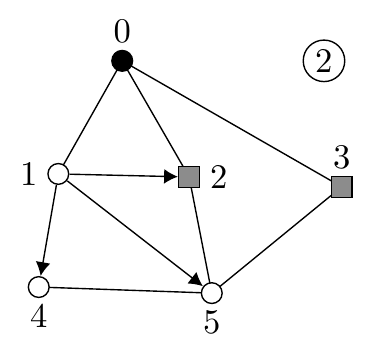}
  \hfill
    \includegraphics[scale=0.9]{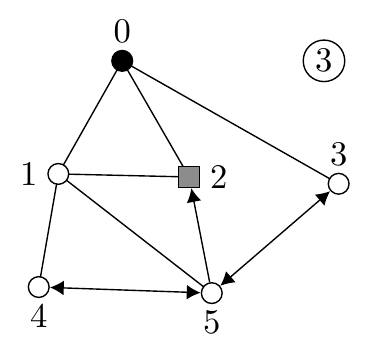}
  \hfill\null

  \hfill
    \includegraphics[scale=0.9]{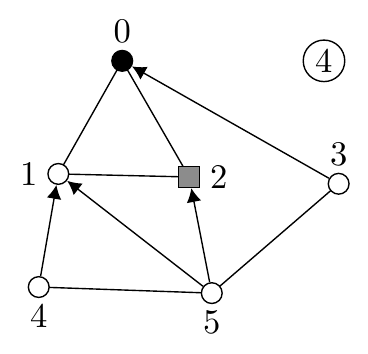}
  \hfill
    \includegraphics[scale=0.9]{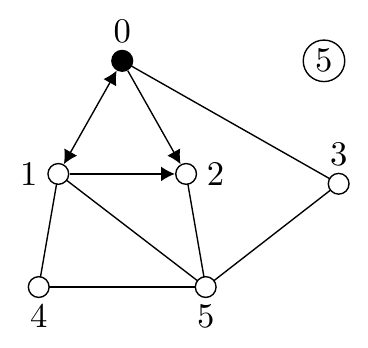}
  \hfill
    \includegraphics[scale=0.9]{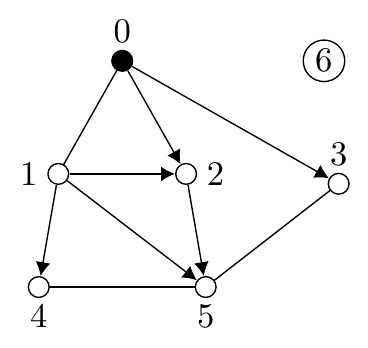}
  \hfill
    \includegraphics[scale=0.9]{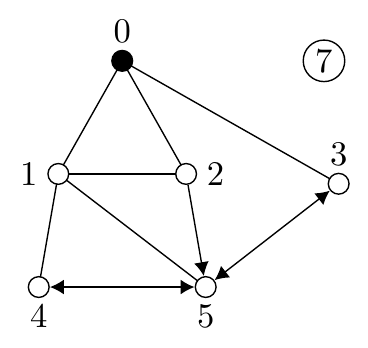}
  \hfill\null

  \hfill
    \includegraphics[scale=0.9]{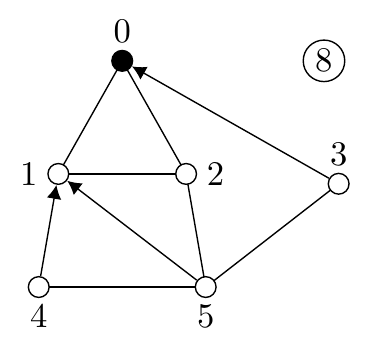}
  \hfill
    \includegraphics[scale=0.9]{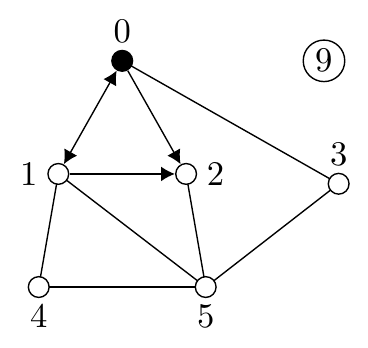}
  \hfill\null
  \caption{A naive extension of algorithm \SynAF does not terminate in
    case of intermittent channel availability. The configuration of
    round $5$ repeats itself in round $9$.}\label{fig:unlimited}
\end{figure}

There is no striking reason for the failure of this naive attempt to
fix \SynAF. To analyze the failure we reconsider the proof of
termination of the original algorithm \SynAF in \cite{Turau:2020b}.
This paper introduces for a given graph $G$ and a broadcasting node
$v_0$ the bipartite auxiliary graph $\G(v_0)$ and shows that
executions of \SynAF on $G$ and $\G(v_0)$ are tightly coupled.
$\G(v_0)$ is a double cover of $G$ that consists of two copies of $G$,
where the cross edges with respect to $v_0$ are removed. Each cross
edges is replaced by two edges leading from one copy of $G$ to the
other. Fig.~\ref{fig:ext_G} depicts $\G(v_0)$ for the graph shown in
Fig.~\ref{fig:unlimited} (see Def.~3 in \cite{Turau:2020b} for
details).

\begin{figure}[h]%
  \hfill
    \includegraphics[scale=0.95]{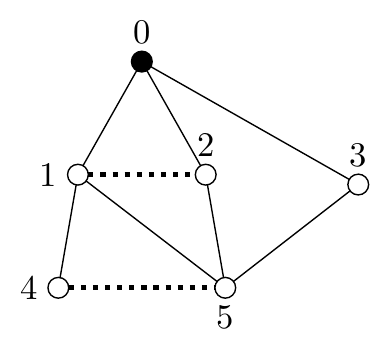}
  \hfill
    \includegraphics[scale=0.95]{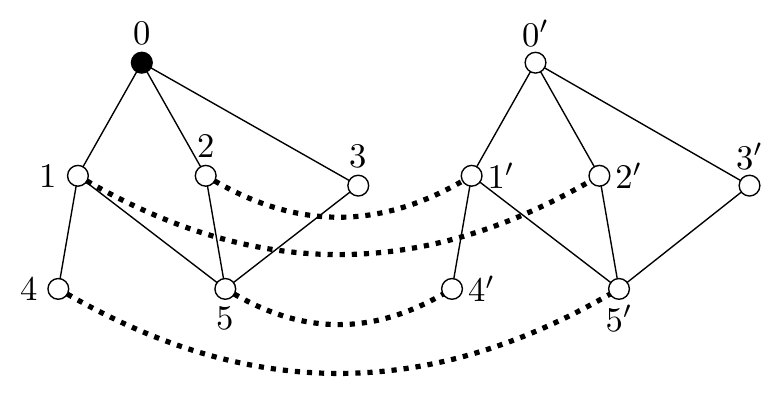}
  \hfill\null
  \caption{The dashed lines on the left show the cross edges of $G$
    ($v_0$ is the broadcasting node). The graph $\G(v_0)$ is shown on
    the right, dashed edges are the replacement
    edges.\label{fig:ext_G}}
\end{figure}

An important observation is that $\G(v_0)$ is bipartite and that in
every round of \SynAF all nodes that send messages belong to one of
the two partitions of nodes. Fig.~\ref{fig:ext_Bipar} shows the
partitioning of the nodes of $\G(v_0)$ for the graph in
Fig.~\ref{fig:ext_G}. An analysis of the execution of
Fig.~\ref{fig:unlimited} shows that in some rounds, nodes from both
partitions forward the message (e.g., in round $3$).

\begin{figure}[h]%
  \hfill
    \includegraphics[scale=0.99]{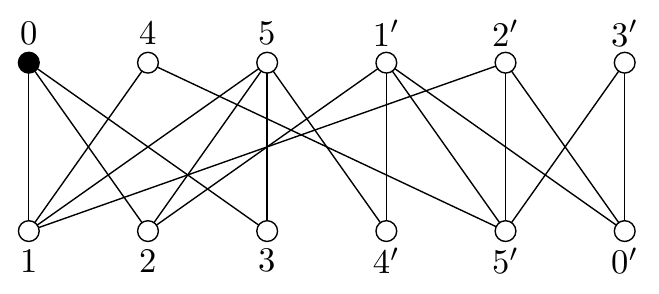}
  \hfill\null
  \caption{Concurrently forwarding nodes in \SynAF either all belong
    the top or bottom row.\label{fig:ext_Bipar}}
\end{figure}

\subsection{Algorithm \SynAFI}
The last observation leads to the following extension of \SynAF for
intermittent availabilities. If a message cannot be forwarded in the
current round, it will be postponed until the next available round
with the same parity, i.e., if the blocked round is odd (resp.\ even),
the message will be forwarded in the next available odd (resp.\ even)
round. This approach guarantees that as in \SynAF all nodes that
concurrently send messages belong to same of the two node sets.
Alg.~\ref{alg:intermitted} shows a realization \SynAFI of this idea.
Compared to \SynAF the new algorithm maintains two sets for the
senders of the message in the variable $M$, one for messages that
arrive in odd rounds and one for even rounds. The parity is maintained
by the Boolean variable $\pr$. The initialization of $\pr$ does not
need be the same for all nodes. The symbol $\bot$ is used to indicate
that no message has arrived in rounds with the specified parity. This
is needed to distinguish this situation from the case that a node
wants to broadcast a message, in this case $M(\pr)$ is assigned the
empty set. If we insert a node $w$ into $M(\pr)$ when $M(\pr)=\bot$
then $M(\pr)=\{w\}$ afterwards. Messages sent in round $i$ are
received in round $i+1$. Hence, in round $1$ no message is received.

\SetProgSty{}
 \begin{algorithm}
   Initialization\\
   ~~~~~\pr := true\;
   ~~~~~$M(\true) := M(\false) := \bot$\;
   \BlankLine
   \BlankLine
   Upon receiving message $m$ from $w$:\\
   ~~~~~$M(\pr).add(w)$\;
  \BlankLine

  \BlankLine
  \If{channel is available and $M(\pr)\not=\bot$} {
  \lForAll{$u\in N(v)\setminus M(\pr)$}{send($u,m$)}
  $M(\pr) := \bot$\;}
   \BlankLine
   \BlankLine
  At the end of each round\\
   ~~~~~$\pr := \neg \pr$\;
   \BlankLine
   \BlankLine
   {\bf function} broadcast($m$)\\
   ~~~~~$M(\pr) := \emptyset$\;
\caption{Algorithm \SynAFI distributes a message $m$ in the
  graph $G$}\label{alg:intermitted}
\end{algorithm}

Fig.~\ref{fig:ex_intermitted} shows an execution of algorithm \SynAFI
for the graph of Fig.~\ref{fig:unlimited}, given that node $v_2$
(resp.\ $v_3$) cannot send in rounds $2$ to $4$ (resp.\ $2$). The
execution terminates after round $5$, with no indeterminacy the
algorithm would terminate in $4$ rounds (see
App.~\ref{app:exec-algor-synaf}).

\begin{figure}[h]%
  \hfill
    \includegraphics[scale=0.92]{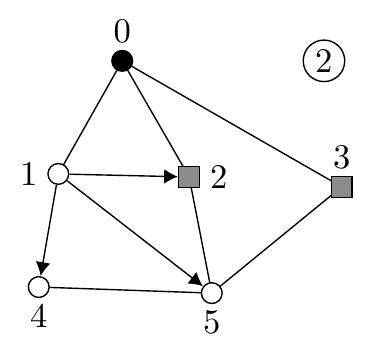}
  \hfill
    \includegraphics[scale=0.92]{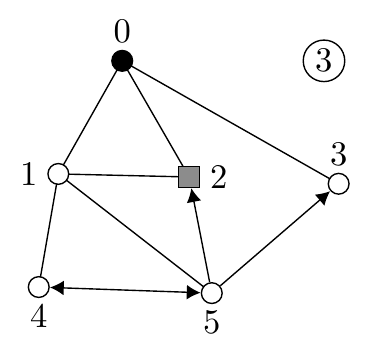}
  \hfill
    \includegraphics[scale=0.92]{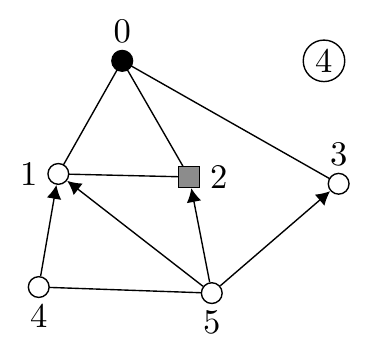}
  \hfill
    \includegraphics[scale=0.92]{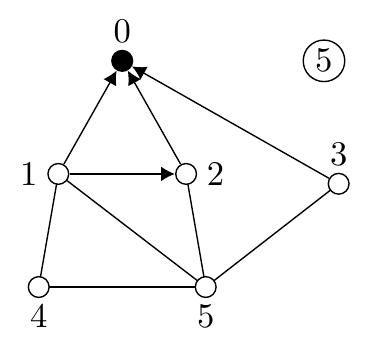}
  \hfill\null
  \caption{Execution of \SynAFI for the graph of
    Fig.~\ref{fig:unlimited}. Round $1$ is the same as in
    Fig.~\ref{fig:unlimited}. In round $6$ node $v_2$ does not need to
    forward the message because, it received messages from all
    neighbors in odd rounds ($1,3,5$). Whereas $v_2$ has to
    send a message to $v_0$ in round $5$ because it only received the
    message from $v_1$ and $v_5$ in even rounds $2$ and
    $4$.}\label{fig:ex_intermitted}
\end{figure}

Clearly this extension of \SynAF is no longer stateless, but because
of message buffering no stateless algorithm can handle channel
unavailabilities.

\subsection{Correctness and Complexity of Algorithm \SynAFI}
\label{sec:proof-corr-synafi}

To formally describe a node's channel availability for message
forwarding the concept of an availability scheme is introduced. Let
$A: V \times \mathbb{N} \longrightarrow \{true,false\}$ be a function.
Node $v$ can send a message in round $c_v$ only if $A(v,c_v) = true$.
$A$ is called an {\em availability scheme} for $G$ and $v_0$ if the
number of pairs $(v,i)\in V \times \mathbb{N}$ with $A(v,i)=\false$ is
bounded by a constant $c$. Note that this concept is only used in the
formal proof. Nodes do not need to have a common round counter. The
availability scheme for Fig.~\ref{fig:unlimited} is
$A(v_2,2)=A(v_2,3)=A(v_2,4)=A(v_3,2)=\false$ and $\true$ otherwise.
WLOG we always assume that $A(v_0,1) = true$.

For a given availability scheme $A$ we construct a directed
bipartite graph $\B_A(v_0)$ such that the execution of \SynAFI on $G$
with respect to $A$ is equivalent to the execution of amnesiac
flooding \SynAF on $\B_A(v_0)$. The starting point for the
construction of $\B_A(v_0)$ is the double cover $\G(v_0)$ of $G$ as
defined in the last section. To keep the notation simple we will omit
the reference to the originating node $v_0$ and refer to the two
graphs as $\B_A$ and \G.

First we extend the definition of the availability scheme $A$ to all
nodes of $\G$, i.e.,
$A: V \cup V' \times \mathbb{N} \longrightarrow \{\true,\false\}$. For
each node $v'\in V'$ let $A(v',i) = A(v,i)$ for all
$i \in \mathbb{N}$. The nodes of $\B_A$ are of two different types:
copies of nodes of \G and so called {\em dummy nodes}. We define
$\B_A$ inductively, layer by layer. There can be copies of the same
node $v$ of \G on several layers of $\B_A$, but the nodes of a single
layer of $\B_A$ are copies of different nodes of \G. Therefore, we do
not cause ambiguity when we denote the copies of the nodes by their
original names. The construction of $\B_A$ is based on a function
$\textit{originator}$, that assigns to each node $v$ of $\B_A$ a set
of neighbors of $v$ in \G. This function is also defined recursively.

Layer $0$ of $\B_A$ consists of copy of $v_0$ with
$\textit{originator}(v_0)=\emptyset$. Layer $1$ consists of copies of
the neighbors of $v_0$ in \G, these are also the neighbors of $v_0$ in
$G$. All layer $1$ nodes are successors of $v_0$ and the
$\textit{originator}$ of these nodes is $\{v_0\}$. Next assume that
layers $0$ to $i$ with $i\ge0$ are already defined including the
function $\textit{originator}$. We first define the nodes of layer
$i+1$ and afterwards the function originator. For each node of layer
$i$ we also define the successors. We do this first for nodes which
are copies of nodes of \G and afterwards for dummy nodes.

Let $v$ be a node of layer $i$ that is a copy of a node of \G. If
$\textit{originator}(v)=N_{\G}(v)$ then $v$ has no successor in layer
$i+1$. Assume $\textit{originator}(v)\not=N_{\G}(v)$. First consider
the case $A(v,i+1)=\textit{true}$. Let
$U = N_{\G}(v) \setminus \textit{originator}(v)$. For each $u \in U$
we do the following: If layer $i+1$ already contains a copy of $u$
then we make it a successor of $v$. Otherwise, we insert a new copy of
$u$ into layer $i+1$ and make it a successor of $v$. If
$A(v,i+1)=\textit{false}$ then we create a new dummy node, insert it
into layer $i+1$, and make it the single successor of $v$. Finally,
let $v$ be a dummy node of layer $i$ and $w$ its single predecessor in
layer $i-1$. If layer $i+1$ already contains a copy of $w$ then we make
it a successor of $v$. Otherwise, we create a new copy of $w$, insert
it into layer $i+1$, and make it the successor of $v$.

To define $\textit{originator}$ for each node $v$ of layer $i+1$ let
$\mathit{pred}(v)$ be the set of predecessors of a node $v$ in
$\cal B$. With $\mathit{pred_d}(v)$ we denote the dummy nodes in
$\mathit{pred}(v)$. Since dummy nodes only have a single predecessor
we denote the predecessor in this case  also by
$\mathit{pred}(v)$. If $v$ is not a dummy node then we define
\[\textit{originator}(v) =  \bigcup_{w \in \mathit{pred}_d(v)}
  \!\!\!\!\!\!\!\!\textit{originator}(w)\,\,\ \cup\,\,\,
  \mathit{pred}(v) \!\setminus\! \mathit{pred}_d(v)\] otherwise
$\textit{originator}(v)= \textit{originator}(\textit{pred}(v))$. Note
that $\B_A$ is bipartite, since nodes of the same layer are not
connected. Fig.~\ref{fig:ext_B} shows the graph $\B_A$ for the graph
of Fig.~\ref{fig:unlimited} and availability scheme $A$. The dummy
nodes are labeled $a$ to $d$. We have
$\textit{originator}(a)=\textit{originator}(b)= \{v_0\}$,
$\textit{originator}(c)=\{v_1\}$, and
$\textit{originator}(d)=\{v_0,v_5\}$. Also,
$\textit{originator}(v_2)=\{v_0,v_5,v_{1'}\}$ in layer $5$.

\begin{figure}[h]%
  \hfill
    \includegraphics[scale=0.99]{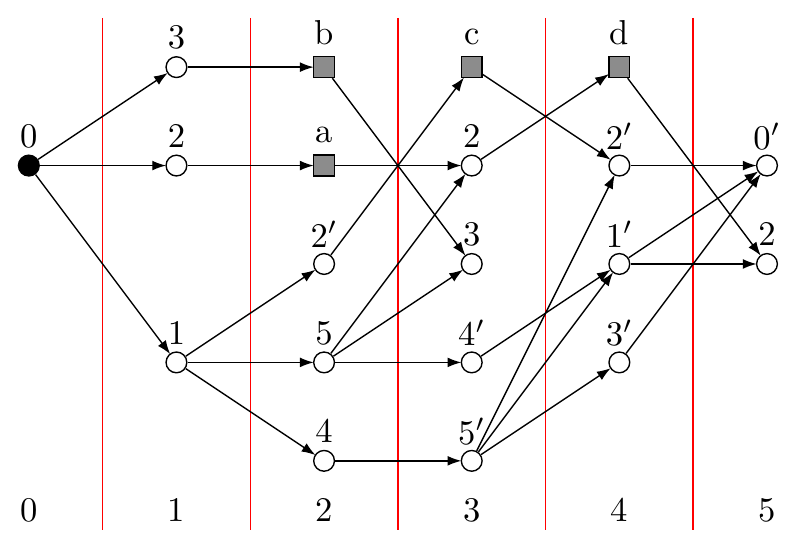}
  \hfill\null
  \caption{The graph $\B_A$ for the availability scheme $A$ has four
    dummy nodes.\label{fig:ext_B}}
\end{figure}

We orient the edges of \G by executing a breadth-first search starting
in $v_0$. The union of the successors and predecessors of a node in \G
are precisely the neighbors of the node in $G$. The next lemma follows
from Lemma 5 of \cite{Turau:2020b}.

\begin{lemma}\label{lem:nonstop}
  Let $v$ be a node of layer $i\ge 0$ of \G. The predecessors of $v$
  in \G are copies of the nodes in $G$ that send in round $i$ of an
  execution of \SynAF a message to $v$ and the successors of $v$ in \G
  receive a message from $v$ in round $i+1$.
\end{lemma}
\begin{proof}
  Suppose that a node $w$ sends in round $i$ a message to a node $v$.
  By Lemma 5 of \cite{Turau:2020b} $w$ is a node of layer $i-1$ and
  either $v$ or $v'$ is a successor of $w$ in \B or $w'$ is a node of
  layer $i-1$ and $v'$ is a successor of $w$. Note that in \B a node
  of $G$ and its copy cannot be in the same layer. The second
  statement also follows from this lemma. 
\end{proof}

Let $A$ be any availability scheme for $G$ and $v_0$.
Lemma~\ref{lem:lem1} is easy to prove.
\begin{lemma}\label{lem:lem1}
  Let $v$ be a node of \G. For each copy $u$ of $v$ in $\B_A$ we have
  $N_{\G(v_0)}(v)= \textit{originator}(u) \cup\textit{succ}(u)$. If
  none of the predecessors of $v$ in \B is a dummy node then
  $N_{\G(v_0)}(v)= \textit{pred}(u) \cup\textit{succ}(u)$.
\end{lemma}

To illustrate the last lemma we consider the execution from
Fig.~\ref{fig:ex_intermitted} and the corresponding graph $\B_A$ in
Fig.~\ref{fig:ext_B}. Let $i=4$ and consider node $v_2$. The copy of
$v_2$ on layer $4$ is called $v_{2'}$. Fig.~\ref{fig:ext_B} shows that
$\textit{originator}(v_2) = \{v_{5},v_1\}$. From
Fig.~\ref{fig:ex_intermitted} we see that node $v_2$ receives a
message from node $v_1$, i.e., $v_1\in v_2.M(\pr)$. Since
$A(v_2,3)=\false$ node $v_2$ could not send a message in round $3$.
Hence the sender $v_1$ of the message received in round $3$ is still
in $v_2.M(\pr)$. This yields $v_2.M(\pr)= \{v_{5},v_1\}$, since
$A(v_2,1)=\true$.

For an availability scheme $A$ and $k\ge 0$ we define a new
availability scheme $A_k$ as follows. We consider the nodes of $\B_A$
in any arbitrary but fixed order and define a total order on the set
of pairs $(v,i)\in V \times \mathbb{N}$ with $A(v,i)=\false$ as
follows: $(v,i) < (w,j)$ if and only if $i < j$ or $i= j$ and $v<w$.
Then we define $A_k(v,i)=\false$ for all but the first $k$ pairs
$(v,i)$, i.e., $A_k$ has value $\false$ for exactly $k$ pairs $(v,i)$.
Note that there exists $c>0$ such that $A=A_{c}$.

\begin{lemma}\label{lem:one-to-one}
  There is a one-to-one mapping between the edges of \G and those
  edges of $\B_A$ that are not incident to a dummy node.
\end{lemma}
\begin{proof}
  It suffices to prove that the lemma holds for each $A_k$ with
  $k\ge 0$. The proof is by induction on $k$. If $k=0$ then the result
  is trivially true since $\B_{A_0}=\G$. Assume the theorem is true
  for $k\ge 1$. Consider the graph $\B_{A_{k-1}}$. Let $(v,i)$ be the
  $k^{th}$ pair with $A(v,i)=\false$. If layer $i-1$ of $\B_{A_{k-1}}$
  contains no copy of $v$ then $\B_{A_{k-1}}=\B_{A_k}$ and we are
  done. Suppose there exists a copy of $v$ on layer $i-1$ of
  $\B_{A_{k-1}}$. We inductively define two sequences of sets $X_j$,
  $\Xbar_j$ ($j\ge 1$) of nodes of $\B_{A_{k-1}}$ (see
  Fig.~\ref{fig:ext}). Nodes of $X_j$, $\Xbar_j$ are in layer $i-1+j$
  of $\B_{A_{k-1}}$. $X_1$ is the set of nodes of layer $i$ that have
  $v$ as the single predecessor in layer $i-1$ and
  $\Xbar_1= \mathit{succ}(v)\setminus X_1$, where $\mathit{succ}(v)$
  denotes the successors in $\B_{A_{k-1}}$. Thus, each node in
  $\Xbar_1$ has besides $v$ another predecessor in layer $i-1$.
  Suppose we already defined $X_{j-1}, \Xbar_{j-1}$. Then $X_j$ is the
  set of nodes of layer $i-1+j$ that have only predecessors in
  $X_{j-1}$, i.e., $\mathit{pred}(X_j) \subseteq X_{j-1}$. $\Xbar_j$
  consists of those nodes of layer $i-1+j$ that have predecessors in
  $X_{j-1}$ and in $\Xbar_{j-1}$, i.e., for each $w\in\Xbar_j$ we have
  $pred(w)\cap X_{j-1}\not=\emptyset$ and
  $pred(w)\cap \Xbar_{j-1}\not=\emptyset$. Hence,
  $\mathit{succ}(X_{j-1}) = X_j\, \dot\cup\, \Xbar_j$. Note that none
  of the nodes of $X_j$, $\Xbar_j$ are dummy nodes, therefore
  $N_{\G}(u)= \textit{pred}(u) \cup\textit{succ}(u)$ for each
  $u\in X_j \cup\Xbar_j$ by Lemma~\ref{lem:lem1}. Since the theorem is
  true for $A_{k-1}$, there exist $t$ such that $X_t=\emptyset$. Note
  that $X_j\not=\emptyset$ for $j=1,\ldots t-1$ while $\Xbar_{j}$ can
  be empty for any $j$.

  \begin{figure}[h]%
  \hfill
    \includegraphics[scale=0.99]{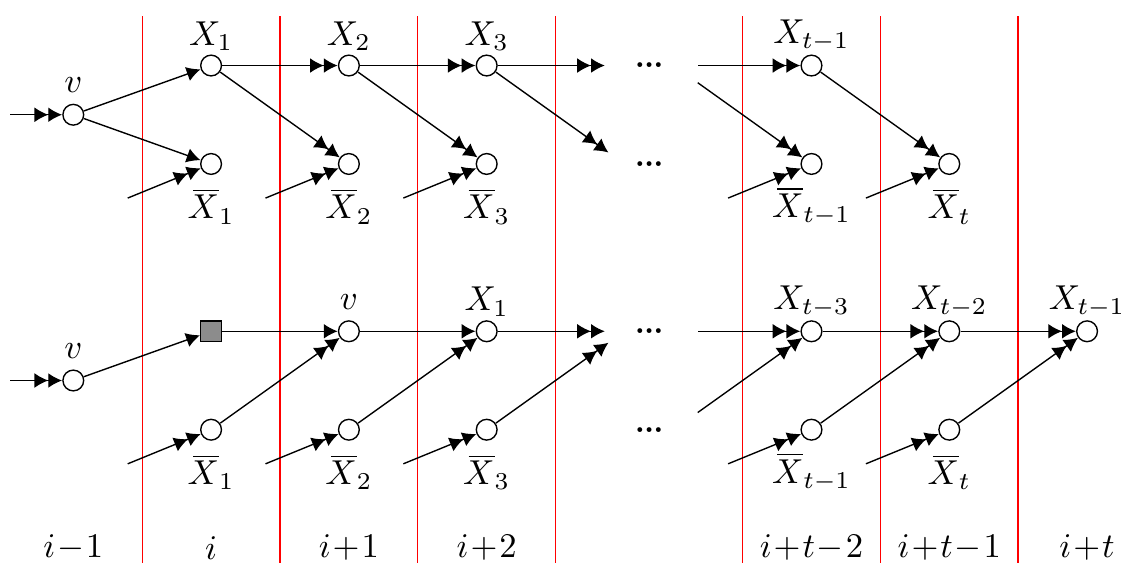}
  \hfill\null
  \caption{The top row illustrates the definition of  $X_{j}$
    and $\Xbar_{j}$ for $\B_{A_{k-1}}$. The lower row displays the
    changes in $\B_{A_{k}}$ compared with $\B_{A_{k-1}}$. The last row
    indicates the number of the layer. The symbol \doublearrow
    indicates that there can be several edges.}\label{fig:ext}
\end{figure}

Next, we show how $\B_{A_k}$ can be derived from $\B_{A_{k-1}}$. The
two graphs coincide completely in the first $i-1$ layers. In
subsequent layers nodes that are not reachable from $v$ in layer $i-1$
also are identical. The single successor of $v$ in layer $i$ is the
dummy node. This node itself has as successor a copy of $v$ on layer
$i+1$. Clearly this copy of $v$ is also the successor of all nodes in
$\Xbar_{1}$ in layer $i$. The successors of the copy of $v$ on layer
$i+1$ are copies of the nodes of set $X_1$. Nodes in $\Xbar_{2}$ on
layer $i+1$ are the predecessors of nodes in $X_1$. All these
statements are an immediate consequence of Lemma~\ref{lem:lem1}.
Similarly it follows that each layer $i-1+j$ for $j\ge 3$ contains
copies of the nodes of set $X_{j-2}$. Their predecessors are copies of
the nodes in $X_{j-3}$ and $\Xbar_{j-1}$.

Thus, in $B_{A_k}$ some edges from $B_{A_{k-1}}$ are reversed: The
orientation of edges from $X_j$ to $\Xbar_{j+1}$ and from $v$ to
$\Xbar_{1}$ is reversed. This analysis also shows that $B_{A_k}$ only
has two additional edges, those adjacent to the new dummy node. In the
worst case, $B_{A_k}$ consists of two more layers compared to
$B_{A_{k-1}}$.
\end{proof}

To ease the formulation of the next lemma we introduce another
definition. Let $u$ be a node of $G$. For a copy of $u$ in layer $i$
of $\B_A$ we denote the originators in $\B_A$ of this copy of $v$ by
$\textit{originator}^{i}(v)$. Furthermore, the set $M(\pr)$ of node $u$
immediately before checking channel availability in round $i$ during
an execution of \SynAFI on $G$ is denoted by $v.M^i(\pr)$.

\begin{lemma}\label{lem:equiv}
  Let $u$ be a non-dummy node of layer $i$ of $\B_A$. Then
  $u.M^i(\pr)=\textit{originator}^i(u)$.
\end{lemma}
\begin{proof}
  We use the notation introduced in the proof of
  Lemma~\ref{lem:one-to-one}. As before we prove by induction on $k$
  that the lemma holds for $A_k$. If $k=0$ then the result holds by
  Lemma~\ref{lem:nonstop} since $\B_{A_0}=\G$. Assume the lemma is
  true for $k\ge 1$. We consider the graph $\B_{A_{k-1}}$. Let
  $(v,i)$ be the $k^{th}$ pair with $A(v,i)=\false$. If in layer $i-1$
  of $\B_{A_{k-1}}$ there exists no copy of $v$ then
  $\B_{A_{k-1}}=\B_{A_k}$ and we are done. Suppose there exists a copy
  of $v$ on layer $i-1$ of $\B_{A_{k-1}}$. From Fig.~\ref{fig:ext} we
  see that we only have to consider the cases $u=v$, $u\in X_{j}$, and
  $u\in \Xbar_{j}$. Remember that there are no dummy nodes in $X_j$,
  $\Xbar_j$.

  First consider the case that $u$ is the copy of $v$ in layer $i+1$
  in $\B_{A_k}$ (see Fig.~\ref{fig:ext}). In round $i+1$ in $\B_{A_k}$
  the nodes in $\Xbar_{1}$ do not receive the message from $v$ because
  $A(v,i)=\false$. Since each node in $\Xbar_{1}$ still receives the
  message from another node, each of them must forward the message in
  round $i$ to $v$. Hence,
  $v.M^{i+1}(\pr) = v.M^{i-1}(\pr) \cup \Xbar_{1}$. On the other hand
  $\textit{originator}^{i+1}(v)=\textit{originator}^{i-1}(v) \cup
  \Xbar_{1}$. By induction
  $\textit{originator}^{i-1}(v)=v.M^{i-1}(\pr)$.


  Next consider the case $u\in X_{1}$. Then $u$ is on layer $i+2$ of
  $\B_k$. Since in $\B_{A_{k-1}}$ each node in $X_{1}$ receives in
  round $i$ only the message from $v$, node $v$ sends the message to
  each node in $X_{1}$ in round $i+1$. Furthermore, since for
  $\B_{A_{k-1}}$ each node in $\Xbar_{2}$ received in round $i+1$ a
  message from a node in $X_1$, each node of $\Xbar_{2}$ sends
  $\B_{A_{k}}$ the message to at least one node of $X_1$. In
  particular node $u$ receives in round $i+2$ the message from its
  predecessors in $\Xbar_{2}$ for $\B_{A_{k}}$. Clearly, $u$ does not
  receive the message from any other node. Thus,
  $u.M^{i+2}(\pr) =\textit{originator}^{i+2}(u)$. The cases
  $u\in X_{j}$ with $j>1$ and $u \in \Xbar_{j}$ with $j \ge 1$ can be
  proved similarly.
\end{proof}

\begin{lemma}\label{lem:depth}
  During round $i$ of an executing of \SynAFI on $G$ a node $v$ sends
  the message to a neighbor $w$ if and only if the copy of $v$ in
  layer $i-1$ of $\B_A$ is the predecessor of a copy of $w$ in layer
  $i$ of $\B_A$.
\end{lemma}
\begin{proof}
  If during the execution of \SynAFI node $v$ sends messages in round
  $i$ to $w$ then $A(v,i)=\true$ and $N(v)\not=v.M(\pr)$. By the
  Lemma~\ref{lem:equiv} we have
  $w\in N(v)\setminus \textit{originator}(v)$. Thus, by construction
  $w$ is a successor of $v$ in $\B_A$. Conversely, if $w$ is successor
  of $v$ in $\B_A$ then $A(v,i)=\true$ and
  $v\in N(v) \setminus \textit{originator}(v)$. Again
  Lemma~\ref{lem:equiv} gives the desired result.
\end{proof}

The last lemma implies that executing \SynAFI on $G$ is {\em
  equivalent} to executing \SynAF on $\B_A$. The reason is that $\B_A$
is bipartite and executing \SynAF on a bipartite graph starting at the
root is equivalent to synchronous flooding the bipartite graph. This
is formulated in the following theorem.

 \begin{theorem}\label{theo:main1}
   Let $G$ be a graph $G$ and $A$ an availability scheme for $G$. Let
   $f=\abs{\{(v,i)\suchthat A(v,i)=\false\}}$. Algorithm \SynAFI
   delivers a broadcasted message (resp.\ terminates) after at most
   $Diam(G)+ 2f$ (resp.\ $2 Diam(G)+ 2f + 1$) rounds. If $G$ is
   bipartite each message is forwarded $\abs{E}$ times, otherwise
   $2\abs{E}$ times.
 \end{theorem}
 \begin{proof}
   Lemma~\ref{lem:depth} implies that \SynAFI terminates after $d$
   rounds where $d$ is the height of $\B_A$. The proof of
   Lemma~\ref{lem:one-to-one} shows that each pair $(v,i)$ with
   $A(v,i)=\false$ increases the depth by at most $2$. By Thm.~1 of
   \cite{Turau:2020b} the depth of $\G$ is at most $2Diam(G)+1$. By
   Lemma~\ref{lem:one-to-one} and Lemma~\ref{lem:equiv} \SynAFI sends
   $2\abs{E}$ messages.
 \end{proof}

\section{Multi-Source Broadcasts}
\label{sec:multi-source}
A variant of broadcasting is multi-source broadcasting, where several
nodes invoke a broadcast of the same message, i.e., with the same
message id, possibly in different rounds. This problem is motivated by
disaster monitoring: A distributed system monitors a geographical
region. When multiple nodes detect an event, each of them broadcasts
this information unless it has already received this information.
Multi-source broadcasting for the case that all nodes invoke the
broadcast in the same round was already analyzed in
\cite{Turau:2020b}. This variant can be reduced to the case of single
node invoking the broadcast by introducing a virtual source $v^\ast$
connected by edges to all broadcasting nodes.

In this section we consider the general case where nodes can invoke
the broadcasts in arbitrary rounds. First we show that broadcasting
one message with algorithm \SynAF also terminates in this case and that
overlapping broadcasts complement each other in the sense that the
message is still forwarded only $2\abs{E}$ resp. $\abs{E}$ times.
Later we extend this to the case of intermittent channels.

\begin{theorem}\label{thm:multi}
  Let $v_1,\ldots, v_{k}$ be nodes of $G$ that broadcast the same
  message $m$ in rounds $r_1,\ldots,r_{k}$. Each broadcast is invoked
  before $m$ reaches the invoking node. Algorithm \SynAF delivers $m$
  after $Diam(G)$ rounds and terminates after at most $2Diam(G)+1$
  rounds and $m$ is forwarded at most $2\abs{E}$ times.
\end{theorem}
\begin{proof}
  WLOG we assume $r_1=0$. For each $i$ with $r_i>0$ we attach to node
  $v_i$ a path $P_i=u^i_1,\ldots, u^i_{r_i}$ with $r_i$ nodes, i.e.,
  $u^i_{r_1}$ is connected to $v_i$ by an edge. The extended graph is
  called $G^\circ$. Let
  $S=\{u^i_1\suchthat r_i>0\} \cup \{v_i\suchthat r_i=0\}$. If in
  $G^\circ$ all nodes in $S$ broadcast in round $0$ message $m$ then
  in round $r_i+1$ each node $v_i$ sends $m$ to all its neighbors in
  $G$. Thus, the forwarding of $m$ along the edges of $G$ is identical
  in $G$ and $G^\circ$. By Thm.~1 of \cite{Turau:2020b} algorithm
  \SynAF delivers $m$ after $d_{G^\circ}(S,V^\circ)$ rounds and
  terminates after at most $d_{G^\circ}(S,V^\circ)+1+Diam(G^\circ)$
  rounds, $V^\circ$ is the set of nodes of $G^\circ$. Also, in
  $G^\circ$ message $m$ is forwarded at most twice via each edge.
  Thus, in $G$ message $m$ is forwarded at most $2\abs{E}$ times.

  To prove the upper bounds for the delivery and termination time we
  reconsider the proof of Thm.~1 of \cite{Turau:2020b}. This proof
  constructs from $G^\circ$ a new graph $G^\ast$ by introducing a new
  node $v^\ast$ and connecting it to all nodes in $S$. It is then
  shown that the termination time of invoking the broadcast in
  $G^\circ$ by all nodes of $S$ in round $0$ is bounded by $d-1$,
  where $d$ is the depth of the bipartite graph $\G(v^\ast)$
  corresponding to $G^\ast$. Note that we are only interested in the
  termination time of the nodes of $G$ in $G^\circ$. Thus, we only
  have to bound the depth of the copies of the nodes of $G$ in
  $\G(v^\ast)$. Since broadcasts are invoked before $m$ is received
  for the first time we have $r_i\le \textit{ecc}_G(v_1)$. Thus, the
  depth of the first copy of each node has depth at most
  $\textit{ecc}_G(v_1)+1\le Diam(G)+1$ in $\G(v^\ast)$. Hence,
  delivery in $G$ takes place after $ Diam(G)$ rounds. The second copy
  of each node of $G$ is at most in distance $1+Diam(G)$ from one of
  the first copies of the nodes of $G$ in $\G(v^\ast)$. Thus,
  termination in $G$ is after at most $2Diam(G)+1$ rounds.
\end{proof}

The stated upper bounds are the worst case. Depending on the locations
of the nodes $v_i$ and the values of $r_i$ the actual times can be
much smaller. Next we extent Thm.~\ref{thm:multi} to tolerate
intermittent channel availabilities.

\begin{theorem}\label{thm:multi_gen}
  Let $A$ be an availability scheme for a graph $G$. Let
  $v_1,\ldots, v_{k}$ be nodes of $G$ that broadcast the same message
  $m$ in rounds $r_1,\ldots,r_{k}$. Each broadcast is invoked before
  $m$ reaches the invoking node. Algorithm \SynAFI delivers $m$
  (resp.\ terminates) in at most $Diam(G)+2f$ (resp.\ $2Diam(G)+2f+1$)
  rounds after the first broadcast with
  $f=\abs{\{(v,i)\suchthat A(v,i)=\false\}}$. Message $m$ is forwarded
  at most $2\abs{E}$ times.
 \end{theorem}
\begin{proof}
  In the proof of Thm.~\ref{thm:multi} it is shown that broadcasting
  the same message $m$ in different rounds by different nodes is
  equivalent to the single broadcast of $m$ by a single node $v^\ast$
  in the graph $G^\ast$. Applying Thm.~\ref{theo:main1} to $G^\ast$
  and $v^\ast$ shows that \SynAF delivers $m$ to all nodes of $G^\ast$
  for any availability scheme. Hence, Thm.~\ref{thm:multi} also holds
  for any availability scheme. 
\end{proof}

\section{Multi-Message Broadcasts}
\label{sec:dynamic}

While algorithm \SynAFI is of interest on its own, it can be used as a
building block for more general broadcasting tasks. In this section we
consider {\em multi-message broadcasts}, i.e., multiple nodes initiate
broadcasts, each with its own message, even when broadcasts from
previous initiations have not completed. We consider this task under
the restriction that in each round each node can forward at most $b$
messages to each of its neighbors. Without this restriction we can
execute one instance of \SynAF for each broadcasted message. Then each
messages is delivered (resp.\ the broadcast terminates) in
$\textit{ecc}(v_0)$ (resp.\ $\textit{ecc}(v_0) + 1 + Diam(G)$) rounds
\cite{Turau:2020b}. The restriction enforces that only $b$ instances
of \SynAF can be active in each round, additional instances have to be
suspended. First consider the case $b=1$.

Multi-message broadcast can be solved with an extension of
algorithm \SynAFI. We use an associative array \mess to store the
senders of suspended messages according to their parity. Message
identifiers are the keys, the values correspond to variable $M$ of
Alg.~\ref{alg:intermitted}. Any time a node $v$ receives a message $m$
with identifier $id$ from a neighbor $w$ it is checked whether
$v.\mess$ already contains an entry with key $id$ for the current
$\pr$. If not, a new entry is created. Then $w$ is inserted according
to the actual value of $\pr$ into $v.\mess[id]$. When all messages of
a round are received all values in $v.\mess$ with the current $\pr$
are checked, if a value equals $N(v)$ then it is set to $\bot$. In
this case $v$ received message $id$ from all neighbors and no action
is required. After this cleaning step, an entry of $\mess$ is selected
for which the value with the current $\pr$ is not $\bot$. Selection is
performed according to a given criterion. The message belonging to
this entry is sent to all neighbors but those listed in the entry.
Finally the entry is set to $\bot$. The details of this algorithm can
be found in App.~\ref{app:dynamic_app}. The delivery order of messages
depends on the selection criterion. The variant of this algorithm
which always selects the method with the smallest id is called
\SynAFIS.

\begin{theorem}\label{theo:dyn_mess}
  Algorithm \SynAFIS eventually delivers each message of any sequence
  of broadcasts of messages with different identifiers. If $G$ is
  bipartite, each message is forwarded $\abs{E}$ times, otherwise $2\abs{E}$ times.
\end{theorem}
\begin{proof}
  The message with the smallest identifier $id_1$ is always forwarded
  first by \SynAFIS. Thus, this message is forwarded as in amnesiac
  flooding. Hence, it is delivered after at most $2Diam(G)+1$ rounds
  after it is broadcasted \cite{Turau:2020b}. Next we define an
  availability scheme $A_1$: $A_1(v,i)=\textit{false}$ if during round
  $i$ of algorithm \SynAFIS node $v$ forwards message $id_1$,
  otherwise let $A_1(v,i)=\textit{true}$. Then the message with the
  second smallest identifier $id_2$ is forwarded as with algorithm
  \SynAFI for availability scheme $A_1$. Thus, by
  Thm.~\ref{theo:main1} this message is eventually delivered. Next
  define availability scheme $A_2$ similarly to $A_1$ with respect to
  the messages with ids $id_1$ and $id_2$ and apply again
  Thm.~\ref{theo:main1}, etc.
\end{proof}

Forwarding the message with the smallest id is only one option. Other
selection criteria are also possible, but without care starvation can
occur. The variant, where the selection of the forwarded message is
fair, is called \SynAFIF. Fairness in this context means, that each
message is selected after at most a fixed number of selections. This
fairness criteria limits the number of concurrent broadcasts. If
message selection is unfair for one of the nodes, then continuously
inserting new messages results in starvation of a message. We have the
following result.

  
\begin{theorem}\label{thm:diffId}
  If in each round each node can forward only one message to each of
  its neighbors algorithm Algorithm \SynAFIF eventually terminates and
  delivers each message of any sequence of broadcasts of messages with
  different identifiers. If $G$ is bipartite, each message is
  forwarded $\abs{E}$ times, otherwise $2\abs{E}$ times.
\end{theorem}
\begin{proof}
  Whenever the associative array \mess of a node is non-empty, the
  node will forward a message in the next round with the adequate
  parity. The fairness assumption implies that whenever $m$ is
  inserted into $w.\mess$ for a node $w$ then after a bounded number
  of rounds it will be forwarded and removed from $w.\mess$. Thus, the
  forwarding of $m$ makes progress.

  Let $m$ be a fixed message that is broadcasted in some round $i_m$.
  Denote by $f_j$ the number of forwards of message $m$ up to round
  $j$. For each $j$ we define an availability scheme $A_j$ as follows:
  $A_j(v,i)=\true$ for all $i>j$ and all $v\in V$. Furthermore,
  $A_j(v,i)=\true$ for $i\le j$ and $v\in V$ if during round $i$ node
  $v$ forwards message $m$. For all other pairs let $A_j(v,i)=\false$.
  Hence, there are only finitely many pairs $(v,i)$ such that
  $A_j(v,i) = \false$. Clearly for all $j$, message $m$ is forwarded
  during the first $j$ rounds as with algorithm \SynAFI with respect
  to $A_j$. Thus, by Thm.~\ref{theo:main1} $f_j\le 2m$. Hence, there
  exist $j_m\ge i_m$ such that in round $j_m$ each node has received
  the message and after this round the message is no longer in the
  system. Hence, the result follows from Thm.~\ref{theo:main1}.
  \end{proof}

  The case $b>1$ is proved similarly. We only have to make a single
  change to \SynAFIF. After the cleaning step we select up to $b$
  entries of \mess and send the corresponding messages. The proof of
  Thm.~\ref{thm:b_mess} is similar to that of Thm.~\ref{thm:diffId}.

  \begin{theorem}\label{thm:b_mess}
    If in each round each node can forward at most $b\ge 1$ messages
    to each of its neighbors algorithm \SynAFIF eventually terminates
    and delivers each message of any sequence of broadcasts of
    messages with different identifiers. If $G$ is bipartite, each
    message is forwarded $\abs{E}$ times, otherwise $2\abs{E}$ times.
  \end{theorem}

  Finally, Thm.~\ref{thm:main} follows directly from
  Thm.~\ref{thm:multi_gen} and Thm.~\ref{thm:b_mess}.

\section{Discussion and Conclusion}
\label{sec:concl}
In this paper we proposed extensions to the synchronous broadcast
algorithm amnesiac flooding. The main extension allows to execute the
algorithm for systems with intermittent channels. While this is of
interest on its own, it is the basis to solve the general task of
multi-message broadcast in systems with bounded channel capacities.
The extended algorithm delivers messages broadcasted by multiple nodes
in different rounds, even when broadcasts from previous invocations
have not completed, while each of the messages is forwarded at most
$2\abs{E}$ times. The main advantage of amnesiac flooding remains,
nodes don’t need to memorize the reception of a message to guarantee
termination.

We conclude by discussing two shortcomings of amnesiac flooding.
\SynAF delivers a broadcasted message twice to each node. To avoid
duplicate delivery, nodes have to use a buffer. Upon receiving a
message $m$ a node checks whether the id of $m$ is contained in its
buffer. If not then $m$ is delivered to the application and $m$'s id
is inserted into the buffer. Otherwise, $m$'s id is removed from the
buffer and not delivered. This also holds for algorithm \SynAFIF.

Amnesiac flooding satisfies the FIFO order, i.e., if a node $v_0$
broadcasts a message $m$ before it broadcasts a message $m'$ then no
node delivers $m'$ unless it has previously delivered $m$. This
property is no longer satisfied for \SynAFI as the following example
shows. Suppose that $v_0$ broadcasts $m$ resp.\ $m'$ in rounds $i$
resp.\ $i+1$. Let $w$ be a neighbor of $v_0$ with $A(w,i+2)=\false$
and $A(v,j)=\true$ for all other pairs. Then node $w$ forwards $m'$ in
round $i+3$ while it forwards $m$ in round $i+4$. Thus, a neighbor $u$
of $w$ receives $m'$ before $m$.

\bibliographystyle{unsrt}
\bibliography{document}

\newpage
\appendix

\section{Execution of Algorithm \SynAF}
\label{app:exec-algor-synaf}

\begin{figure}[h]%
  \hfill
    \includegraphics[scale=0.85]{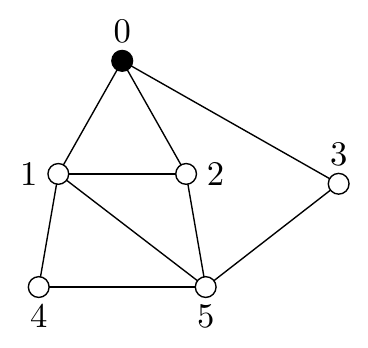}
  \hfill
    \includegraphics[scale=0.85]{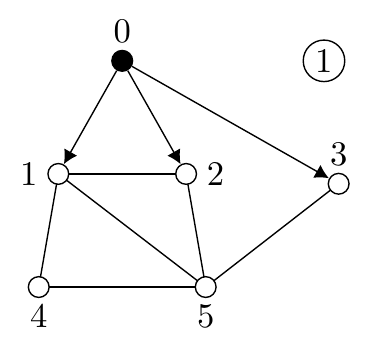}
  \hfill
    \includegraphics[scale=0.85]{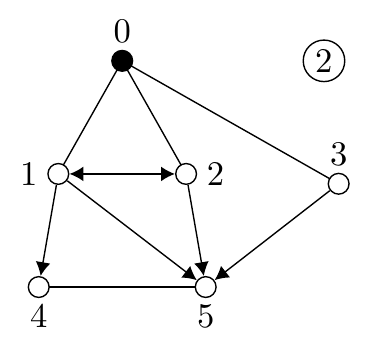}
  \hfill
    \includegraphics[scale=0.85]{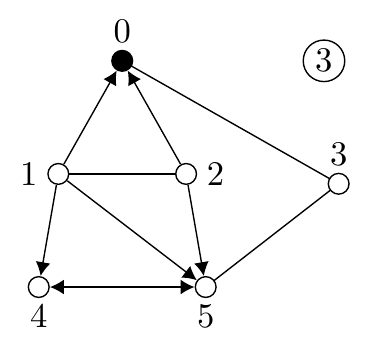}
  \hfill
    \includegraphics[scale=0.85]{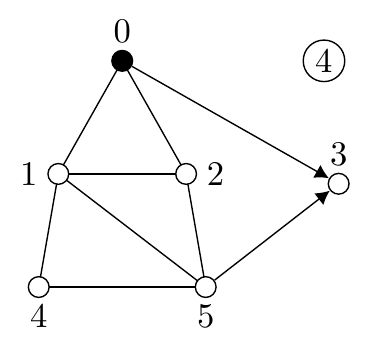}
  \hfill\null
  \caption{The original amnesiac flooding algorithm \SynAF terminates
    after 4 rounds for the graph of
    Fig.~\ref{fig:unlimited}.}\label{fig:XX}
\end{figure}

\section{Algorithm for Multi-Message Broadcast}
\label{app:dynamic_app} 
In this section we describe the extension of algorithm \SynAFI to
realize multi-message broadcasts. As with \SynAFI each node $v$ has
two variables. First, a Boolean flag \pr that is toggled at the end of
every round. The values of \pr must not be synchronized among nodes.
The second variable corresponds to variable $M$ of \SynAF, it is used
to store the senders of the messages according to the parity of the
round in which they were received. In multi-message broadcasts a node
can receive different messages in a round and therefore must be
prepared to separately store the senders of these messages. An
associative array \mess is used for this purpose. Message identifiers
are the keys, the values correspond to variable $M$ of \SynAFI. Values
consist of two parts $list[\true]$ and $list[\false]$, corresponding
to the round's parity. The symbol $\bot$ indicates that no message has
arrived in rounds with the specified parity. This is needed to
distinguish this from the case when a node invokes a broadcast, in
this case the value is the empty set $\emptyset$. If we insert a node
$w$ when the value is $\bot$ then it is $\{w\}$ afterwards.
Tab.~\ref{tab:messTab} shows an example of \mess.

\renewcommand{\arraystretch}{1.25}
\begin{table}[htbp]
  \centering
  \begin{tabular}{c|c|c|c}
Message Id & Message & $list[\true]$ & $list[\false]$ \\ \hline\hline
    $17$ & $\ldots$ &$v_1,v_3$ & $v_1, v_4$ \\ \hline
    $123$ & $\ldots$ &$v_{15}$ & $\bot$ \\ \hline
    $3$ & $\ldots$ &$\emptyset$ & $\emptyset$\\
  \end{tabular}
  \caption[Tabelle]{Example of a node's associative array \mess.}
  \label{tab:messTab}
\end{table}

Fig.~\ref{fig:algo} shows the pseudo code of the proposed extension of
\SynAFI. In every round the following three steps are executed: First,
received messages are used to update the message table. In the second
step a message is selected from the message table and sent to those
neighbors not listed in the appropriate column of the corresponding
row. As a last step the flag \pr is toggled. 

\begin{figure}
  \par\noindent\rule{\textwidth}{0.4pt}
  \input{algo}
  \par\noindent\rule{\textwidth}{0.4pt}
  \caption{Algorithm\label{fig:algo}}
\end{figure}

Next we describe the first two steps at full length. The details of
the first step are as follows. Any time a node $v$ receives a message
$m$ with identifier $id$ from a neighbor $w$ it is checked whether
$v$'s message table already contains a row for $id$. If not, a new row
is created and the first two columns are filled with $id$ and $m$. The
last two columns contain the symbol $\bot$. In any case the node $w$
is appended to the list in the third or forth column according to the
current \pr into $v.\mess[id].list$. In case the
corresponding entry is $\bot$ a new list with the single element $w$
is created.

When all messages of a round are received then the following cleaning
action is performed as the closing-off of the first step. All values
in $v.\mess$ with the current $\pr$ are checked. If a value equals
$N(v)$ then it is set to $\bot$. In this case $v$ received message
$id$ from all neighbors and no action is required. After this cleaning
step, an entry of $\mess$ is selected for which the value with the
current $\pr$ is not $\bot$. Selection is performed according to a
given criterion. The message belonging to this entry is sent to all
neighbors but those listed in the entry. Finally the entry is set to
$\bot$.

Initially for each node $v$ the associative array \mess is empty and
flag \pr has an arbitrary value. A node $v$ that wants to disseminate
a new message $m$ with the identifier $id$ creates a new row in the
message table and inserts the value $id$ and $m$ into the first two
columns. The last two columns contain the empty list $\emptyset$. The
third row of Tab.~\ref{tab:messTab} is an example for this situation.

If the node $v$ with the message table shown in Tab.~\ref{tab:messTab}
receives in a round with $\pr=\false$ a message with $id=17$ from
neighbors $v_1,v_5$, and $v_8$ the last column of the corresponding
row would be updated to $v_1,v_4,v_5,v_8$. If the id of the received
message is $123$ then the last column would be updated to
$v_1,v_5,v_8$. Next we give an example for the execution of the second
part of the algorithm for Tab.~\ref{tab:messTab}. If $\pr =\false$ and
the first row is selected, the message with $id=17$ is sent to all
neighbors of $v$ except $v_1$ and $v_4$. If the last row is selected,
the message with $id=3$ is sent to all neighbors. In the first case
the last column is set to $\bot$ and the row remains in the table. In
the second case the row is deleted. If on the other hand $\pr =\true$
and the second row is selected the message with $id=123$ is sent to
all neighbors of the node except node $v_{15}$ and the row is deleted
from the table.

\end{document}